\documentclass[11pt]{article}
\usepackage{a4wide}
\usepackage{amsthm}
\usepackage{epsfig}
\usepackage{color}
\usepackage{amssymb}
\usepackage{amsmath}

\def\nca{\operatorname{nca}}
\def\succ{\operatorname{succ}}
\def\pred{\operatorname{pred}}

\newtheorem{lemma}{Lemma}
\newtheorem{theorem}{Theorem}

\begin{document}

\title{Approximate Distance Oracles for Planar Graphs with Improved Query Time-Space Tradeoff}
\author{Christian Wulff-Nilsen
        \footnote{Department of Computer Science,
                  University of Copenhagen,
                  \texttt{koolooz@di.ku.dk},
                  \texttt{http://www.diku.dk/$_{\widetilde{~}}$koolooz/}}}

\date{}

\maketitle
\begin{abstract}
We consider approximate distance oracles for edge-weighted $n$-vertex undirected planar graphs. Given fixed $\epsilon > 0$, we present a $(1+\epsilon)$-approximate distance oracle with $O(n(\log\log n)^2)$ space and $O((\log\log n)^3)$ query time. This improves the previous best product of query time and space of the oracles of Thorup (FOCS $2001$, J.~ACM $2004$) and Klein (SODA $2002$) from $O(n\log n)$ to $O(n(\log\log n)^5)$.
\end{abstract}

\section{Introduction}\label{sec:Intro}
Given an $n$-vertex edge-weighted undirected planar graph $G$, a distance oracle for $G$ is a data structure that can efficiently answer distance queries $d_G(u,v)$ between pairs of vertices $(u,v)$ in $G$. One way of achieving this is to simply store an $n\times n$-distance matrix where $n$ is the number of vertices. Each query can be answered in constant time but the space requirement is large.

If one is willing to settle for approximate distances, much more compact oracles exist. It has been shown that for any $\epsilon > 0$, there is a $(1+\epsilon)$-approximate distance oracle for $G$ of size $O(\frac 1\epsilon n\log n)$ which for any query pair $(u,v)$ outputs in time $O(1/\epsilon)$ an estimate $\tilde{d}_G(u,v)$ such that $d_G(u,v)\leq\tilde{d}_G(u,v)\leq(1+\epsilon)d_G(u,v)$ (Thorup~\cite{OraclePlanarThorup} and Klein~\cite{OraclePlanarKlein}).

The oracles of Thorup and of Klein both rely on a recursive decomposition of $G$ using shortest path separators from a shortest path tree $T$: first $G$ is decomposed into two subgraphs with such a separator and then the two subgraphs are recursively decomposed. An important observation is that for any vertex $u$ and any shortest path separator $S$, there is a size $O(1/\epsilon)$ set $P$ of so called portals on $S$ which are vertices such that for any $w\in S$, there exists a $p\in P$ such that $d_G(u,p) + d_S(p,w)\leq (1+\epsilon)d_G(u,w)$. Thus, to get an approximate distance from $u$ to any $w\in S$ only $O(1/\epsilon)$ distances $d_G(u,p)$ need to be stored in addition to distances in $T$. The oracle stores distances from $u$ to portals on each of the $O(\log n)$ separators above $u$ in the recursive decomposition tree, giving a total space of $O(\frac 1 \epsilon n\log n)$. To answer a $uv$-query, the oracle identifies the nearest common ancestor separator $S_{uv}$ of $u$ and $v$ in the recursive decomposition. As $S_{uv}$ separates $u$ and $v$, distances from $u$ and from $v$ to their respective portal sets on $S_{uv}$ can be combined to obtain a $(1+\epsilon)$-approximate distance estimate in $O(1/\epsilon)$ time.

Additional oracles for planar graphs have since been presented. Kawarabayashi, Klein, and Sommer~\cite{LinSpaceOraclesPlanar} showed how to improve space to $O(n)$ at the cost of an increase in query time to $O(\frac 1{\epsilon^2}\log^2n)$ and gave generalizations to bounded-genus and minor-free graphs. Kawarabayashi, Sommer, and Thorup~\cite{CompactOraclesPlanar} focused on improving the space-query time tradeoff and gave an oracle with $\overline{O}(n\log n)$ space and $\overline{O}(1/\epsilon)$ query time, where $\overline{O}(\cdot)$ hides $\log\log n$ and $\log(1/\epsilon)$ factors, thereby essentially improving the query time-space product from $O(\frac 1{\epsilon^2}n\log n)$ to $O(\frac 1\epsilon n\log n)$. They also showed that if the average edge weight is poly-logarithmic, $\overline{O}(n)$ space and $\overline{O}(1/\epsilon)$ query time can be obtained.

Except for planar graphs with poly-logarithmic average edge weights, every oracle presented so far has a query time-space product of order $\Theta(n\log n)$ (ignoring the dependency on $\epsilon$). We finally break this barrier by giving an oracle with $O(n(\log\log n)^2)$ space and $O((\log\log n)^3)$ query time. The exact bounds are given in the following theorem.
\begin{theorem}\label{Thm:Main}
Let $G$ be an $n$-vertex undirected edge-weighted planar graph. For any $0 < \epsilon < 1$, there is a $(1+\epsilon)$-approximate distance oracle of $G$ with query time $O((\log\log n)^3/\epsilon^2 + \log\log n\sqrt{\log\log((\log\log n)/\epsilon^2)}/\epsilon^2)$ and space $O(n((\log\log n)^2/\epsilon + (\log\log n)/\epsilon^2))$.
\end{theorem}
Our dependency on $\epsilon$ in the query time-space product is worse than in~\cite{OraclePlanarKlein, OraclePlanarThorup} but still only a low-degree polynomial in $1/\epsilon$; it is roughly $1/\epsilon^3$ when $1/\epsilon = O(\log\log n)$ and roughly $1/\epsilon^4$ otherwise, compared to $1/\epsilon^2$ in~\cite{OraclePlanarKlein, OraclePlanarThorup}. Focus in this paper is on improving the dependency on $n$ and not $\epsilon$ which we regard as fixed. Our data structure uses randomization due to hashing and fast integer sorting. Space and query time for hashing can be made worst-case with expected construction time. For sorting, we use the algorithm of Han and Thorup~\cite{IntegerSorting} to get the bound in Theorem~\ref{Thm:Main}. To make our data structure deterministic, we can use a standard optimal comparison sort or the slightly faster deterministic integer sorting algorithm of Fredman and Willard~\cite{DetIntSorting}. With the latter, we get a deterministic query time of our data structure of $O((\log\log n)^3/\epsilon^2 + (\log\log n\log((\log\log n)/\epsilon^2))/(\epsilon^2\log\log((\log\log n)/\epsilon)))$.

A main difference between our oracle and those of Thorup and of Klein is that we do not store distances from each vertex $u$ to portals on all $O(\log n)$ separators above $u$. Instead we save space by introducing a shortcutting system to the recursive decomposition tree so that we can get from $u$ to any separator above it using $O(\log\log n)$ shortcuts. Each shortcut corresponds to a pair of separators on a root-to-leaf path in the tree and we essentially store approximate distances from vertices on one separator to portals on the other separator closer to the root of the recursive decomposition tree. This complicates the query algorithm and its analysis since approximate distances are found in $O(\log\log n)$ steps instead of just one.

\subsection{Related work}
Thorup~\cite{OraclePlanarThorup} also gave an oracle for planar \emph{digraphs} which for polynomially bounded edge weights achieves $O(\frac 1\epsilon n\log^2n)$ space and close to $O(1/\epsilon)$ query time. Exact oracles for planar digraphs with tradeoff between space and query time have been studied but require near-quadratic space for constant or near-constant query time~\cite{ExactOraclePlanarMozesSommer,CWNPHD}.

For a general undirected $n$-vertex graph $G$, Thorup and Zwick showed that for any parameter $k\in\mathbb N$, there is a $(2k-1)$-approximate distance oracle for $G$ with $O(kn^{1+1/k})$ space and $O(1/k)$ query time which is believed to be essentially optimal due to a girth conjecture of Erd\H{o}s~\cite{Erdos}. Variations and slight improvements have since been presented; see, e.g.,~\cite{ChechikOracleGeneral,MendelNaor,PatrascuRoditty,CWNOracleGeneral1,CWNOracleGeneral2}.

\subsection{Organization of the paper}
In Section~\ref{sec:Prelim}, we give some basic definitions, notation, and a variant of a standard recursive decomposition of planar graphs with shortest path separators. We then present our oracle in Sections~\ref{sec:PhaseI} and~\ref{sec:PhaseII}. Section~\ref{sec:PhaseI} presents the first phase of the query algorithm. This phase computes approximate distances from query vertices $u$ and $v$ to certain portals on the nearest common ancestor separator $S_{uv}$ of $u$ and $v$ in the recursive decomposition tree but where approximate shortest paths are restricted to the child regions of $S_{uv}$ in the recursive decomposition. The second phase in Section~\ref{sec:PhaseII} uses the output of the first phase to then find approximate distances from $u$ and $v$ to portals on $S_{uv}$ in the entire graph $G$. A main challenge is that the vertices of $S_{uv}$ are not represented explicitly on all recursion levels but on various levels on the path from $S_{uv}$ to the root of the recursive decomposition; the second phase traverses this path to find the desired approximate distances and portals. From the output of Phase II, obtaining an approximate $uv$-distance can then be done efficiently, as we show in Section~\ref{sec:ApproxDist}. Finally, we make some concluding remarks in Section~\ref{sec:ConclRem}.

\section{Preliminaries}\label{sec:Prelim}
For a graph $G$, denote by $V(G)$ and $E(G)$ its vertex set and edge set, respectively. When convenient, an edge $(u,v)$ with weight $w$ is denoted $(u,v,w)$. For a rooted tree $T$ and two nodes $u,v\in T$, denote by $\nca_T(u,v)$ the nearest common ancestor of $u$ and $v$ in $T$. For a path $P$ and for two vertices $u,v\in P$, $P[u,v]$ denotes the subpath of $P$ between $u$ and $v$. As in previous papers on approximate distance oracles for planar graphs, we assume the Word-RAM model with standard instructions.

\subsection{Recursive Decomposition}\label{subsec:RecDecomp}
In the following, $G = (V,E)$ denotes an $n$-vertex, undirected, edge-weighted planar embedded graph and $T$ is a shortest path tree in $G$ rooted at a source vertex $s$. By performing vertex-splitting, we may assume that $G$ has degree three.

The oracle of Thorup keeps a recursive decomposition of $G$ consisting of shortest path separators. Our oracle obtains a similar decomposition but we need it to have some additional properties which we focus on in the following.

Denote by $G_\Delta$ an arbitrary triangulation of $G$ where edges of $E(G_\Delta)\setminus E$ are called \emph{pseudo-edges} and are given infinite weight. A shortest path separator of $G_\Delta$ w.r.t.~an assignment of weights to triangles of $G_\Delta$ consists of a (possibly non-simple) cycle $C$ defined by two shortest paths $s\leadsto u$ and $s\leadsto v$ in $T$ and a non-tree edge $(u,v)$; the total weight of triangles on each side of $C$ is at most $\frac 2 3$ of the total weight of all triangles of $G_\Delta$. See~\cite{LiptonTarjan,OraclePlanarThorup} for details.

First, we decompose $G_\Delta$ into two subgraphs enclosed by $C$; both subgraphs inherit the edges and vertices of $C$, for a suitable weight function on triangles. Degree two vertices $u$ are removed from each subgraph by replacing their incident edges $(v,u)$ and $(u,w)$ with a single edge $(v,w)$ whose weight is the sum of weights of $(v,u)$ and $(u,w)$. Then the two subgraphs are recursively decomposed until constant-size subgraphs are obtained.

For each subgraph $R$ obtained in the above recursive procedure, we form a \emph{region} $R'$ as follows. Subgraph $R$ contains $O(\log n)$ separators of the form $s\leadsto u\rightarrow v\leadsto s$, namely those formed from the root of the recursion down to $R$. These are separators formed from the root of the recursion down to $R$. Region $R'$ is obtained from $R$ by removing pseudo-edges, except those contained in the separators that formed $R$, and then removing degree two vertices as above. Let $H_1,\ldots,H_k$ be the faces of $R$ containing vertices/edges of $G$ not belonging to $R$. Each $H_i$ is a separator $s\leadsto u\rightarrow v\leadsto s$ (possibly with some degree two vertices removed) and we call $H_i$ a \emph{hole} (of $R'$). Ancestor/descendant relations between regions are defined according to their nesting in the recursive decomposition tree which we denote by $\mathcal T$. In Section~\ref{subsec:ConstructRecDecomp}, we show how to pick shortest path separators such that
\begin{enumerate}
\item there are $O(n)$ regions in total each having $O(1)$ holes,
\item for each child $R'$ of each region $R$, every hole of $R$ (regarded as the closed set of the plane inside the hole) is fully-contained in a hole of $R'$,
\item the height of $\mathcal T$ is $O(\log n)$.
\end{enumerate}
When we refer to a recursive decomposition in the following, we assume it has these properties. Observe that for all regions $R$ and $R'$ where $R$ is a descendant of $R'$, $V(R)\subseteq V(R')$. This follows since $R$ is obtained from $R'$ by eliminating subgraphs of $R'$ and degree two vertices.

\subsection{Constructing a recursive decomposition}\label{subsec:ConstructRecDecomp}
In the following, refer to the triangulated subgraphs obtained when recursively decomposing $G_\Delta$ as \emph{$\Delta$-regions}; holes of $\Delta$-regions are defined to be the holes of the corresponding regions with pseudo-edges added to form the triangulation. We now show how to pick the separators so that the following \emph{region conditions} are satisfied:
\begin{enumerate}
\item each $\Delta$-region of even resp.~odd depth in $\mathcal T$ has at most three resp.~four holes,
\item for each $\Delta$-region of even depth in $\mathcal T$, if it contains exactly $f$ faces of $\Delta G$, each of its grandchildren contain at most $2f/3$ faces of $\Delta G$,
\item for each child $R'$ of each $\Delta$-region $R$, every hole of $R$ (regarded as the closed set of the plane inside the hole) is fully-contained in a hole of $R'$.
\end{enumerate}
Furthermore, we show that the number of regions is $O(n)$ and that the height of $\mathcal T$ is $O(\log n)$. This gives the desired properties for a recursive decomposition, as stated above.

All separators are formed from shortest path tree $T$. The recursion stops once $\Delta$-regions with at most two faces of $\Delta G$ are obtained. For a $\Delta$-region $R$ of even depth in $\mathcal T$, assume it has at most three holes (this trivially holds for $G_\Delta$ at the root of $\mathcal T$). We assign a unit of weight to each face of $R$ that is also a face of $G_\Delta$. All other faces of $R$ are given weight $0$. Using the subtree of $T$ in $R$, we find a balanced shortest path separator w.r.t.~this weight function. If $f$ is the number of faces of $\Delta G$ in $R$ then each of the subgraphs formed contain at most $2f/3$ of these faces. Furthermore, each of these subgraphs have at most four holes since at most one new hole is formed when removing one side of the separator.

Now consider a $\Delta$-region $R$ of odd depth in $\mathcal T$ and assume it has at most four holes. If $R$ has at most two holes, we decompose it as described above for even-depth $\Delta$-regions; the sub-$\Delta$-regions formed will have at most three holes. Otherwise, we define a different weight function than that above: for each hole $H$ of $R$, exactly one of the triangles of $R$ contained in $H$ is given unit weight. All other triangles of $R$ are given weight $0$. The shortest path separator w.r.t.~this weight function ensures that each of the two sub-$\Delta$-regions of $R$ formed will have at most three holes, namely at most $\lfloor\frac 2 3\cdot 4\rfloor = 2$ holes inherited from $R$ and one hole formed by removing one side of the separator.

It is now clear that the first region condition holds. The second region condition holds as well from the above and from the fact that each child of an odd-depth $\Delta$-region $R$ cannot contain more faces of $\Delta G$ than $R$.

Assume for the sake of contradiction that the third region condition does not hold for some $\Delta$-region $R$ and one of its children. Then the separator that was used to decompose $R$ must have used one of the pseudo-edges in the triangulated hole $H$ of $R$. Since the boundary of $H$ consists of two shortest paths from $T$ and a single pseudo-edge, one of the children of $R$ must be fully contained in $H$, which for both choices of weight function above gives an unbalanced separator, a contradiction. We conclude that the third region condition holds.

The number of $\Delta$-regions is asymptotically bounded by the number of leaves of the recursive decomposition tree. Each leaf $\Delta$-region has a parent with at least three faces of $\Delta G$. For any two distinct $\Delta$-regions $R$ and $R'$ that are parents of leaf $\Delta$-regions, no face of $\Delta G$ belongs to both $R$ and $R'$. It follows that there are only $O(n)$ parents of leaf $\Delta$-regions. Since the decomposition tree $\mathcal T$ is binary, the total number of leaf $\Delta$-regions is $O(n)$. This implies that the total number of $\Delta$-regions, and hence regions, is $O(n)$.

It follows from the second region condition and the termination condition for the recursion that $\mathcal T$ has height $O(\log n)$.

\subsection{Region boundary structure}
Let $Q$ be a path in $T$ from $s$ to some vertex. For any subpath $Q[u,v]$ of $Q$, let $v_1,\ldots,v_k$ be those interior vertices of $Q$ having an incident edge of $E$ emanating to the left of $Q[u,v]$ when looking in the direction from $u$ to $v$. We order the vertices such that $Q[u,v] = u\leadsto v_1\leadsto\cdots\leadsto v_k\leadsto v$. The \emph{left side} of $Q[u,v]$ is the $uv$-path with edges $(u,v_1), (v_1,v_2),\ldots,(v_{k-1},v_k), (v_k,u)$ where each edge has weight equal to the weight of the corresponding subpath of $Q$. We define the \emph{right side} of $Q[u,v]$ similarly.

For a region $R$, denote by $\delta R$ the \emph{boundary} of $R$ which is the subgraph of $R$ contained in the $O(1)$ holes of $R$. For each hole $H$, it will be convenient to regard the two shortest paths in $T$ bounding $H$ as disjoint in $\delta R$ by replacing one path with its left side and the other with its right side; vertices shared by the original two paths are regarded as distinct in the two new paths, see Figure~\ref{fig:EulerTour}. Note that $\delta R$ represented in this way is now a single face of $R$.
\begin{figure}
\centerline{\scalebox{0.7}{\input{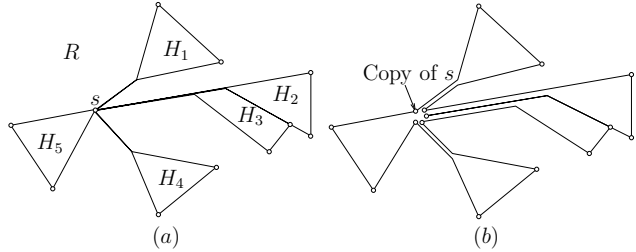}}}
\caption{(a): A region R with $\delta R$ bounding holes $H_1,\ldots,H_5$. Corners are shown as white vertices. (b): Representation of $\delta R$ obtained by ``cutting open'' each hole.}
\label{fig:EulerTour}
\end{figure}
Cutting open each hole like this ensures that paths in $R$ do not cross its boundary.

A vertex (edge) of $R$ that is not contained in $\delta R$ is called \emph{interior}. Denote by $\mathcal P_R$ the $O(1)$ shortest paths from the shortest path tree $T$ that bound $\delta R$. A vertex which is an endpoint of a path in $\mathcal P_R$ is called a \emph{corner} of $R$; see Figure~\ref{fig:EulerTour}. We let $C(R)$ denote the set of $O(1)$ corners of $R$. In some places, we will instead consider the face $\delta_G R$ obtained from $\delta R$ by replacing each edge with its corresponding path in $G$. The following lemma gives the structure of $\delta R$.
\begin{lemma}\label{Lem:deltaR}
For any region $R$, $\delta R$ consists of $O(1)$ subpaths each of which is either a single (possibly pseudo-)edge between two corners, an edge corresponding to a path in $T$ from $s$ to some corner of $R$, or the left or right side of a path in $T$.
\end{lemma}
\begin{proof}
Consider partitioning $\delta R$ into subpaths each starting and ending in a corner and with no interior corner vertices. There are only $O(1)$ such subpaths and each of them is either an edge or pseudo-edge ending in corners or the subpath corresponds to a path in $T$ incident to one or two holes. If it is incident to one hole, it is the left or right side of a path in $T$ and if it is incident to two holes, it is a single edge corresponding to a path ending in $s$.
\end{proof}

When convenient, we identify regions with their corresponding nodes in $\mathcal T$. For two regions $R$ and $R'$, denote by $R\leadsto R'$ the simple path from $R$ to $R'$ in $\mathcal T$.

\subsection{Portals}
Let $G = (V,E)$ be an undirected edge-weighted (not necessarily planar) graph, let $u\in V$, and let $Q$ be a shortest path in $G$. Thorup~\cite{OraclePlanarThorup} showed that for any given $\epsilon > 0$, $V(Q)$ contains a set $P$ of $O(1/\epsilon)$ \emph{portals} such that for any $v\in V(Q)$, there is a $p\in P$ such that $d_G(u,p) + d_P(p,v)\leq (1+\epsilon)d_G(u,v)$.

Let $R$ be a region and let $P\subseteq V(\delta R)$. Given a value $\varepsilon > 0$, a vertex $u\in V$, and an undirected (possibly non-planar) graph $H$ with $P\cup\{u\}\subseteq V(H)\subseteq V$ where every edge $(v_1,v_2)$ in $H$ corresponds to a path $v_1\leadsto v_2$ in $G$ of the same weight. Then $P$ is a \emph{$(u,H,1+\varepsilon)$-portal set of $\delta R$} if for any $v\in V(\delta R)$, there exists a vertex $p\in P$ such that $d_H(u,p) + d_{\delta R}(p,v) \leq (1+\varepsilon)d_{H\cup\delta R}(u,v)$. We call $p$ a \emph{portal} (of $P$). Applying the portal construction of Thorup~\cite{OraclePlanarThorup} (see also the proof of Lemma~\ref{Lem:Portal}) to each path in $\mathcal P_R$ gives a $(u,H,1+\varepsilon)$-portal set of $\delta R$ of size $O(|\mathcal P_R|/\varepsilon) = O(1/\varepsilon)$. Define $(u,H,1+\varepsilon)$-portal set of $\delta_G R$ similarly; its size is $O(1/\varepsilon)$ as well.

We need a slightly more general result regarding portals than that of Thorup which we state in the following somewhat technical lemma. It roughly says that if we have a graph $H$ representing some subgraph of $G$ such that distances in $G$ from a vertex $u$ to a shortest path $Q$ in $G$ are approximated well in $H$, then we can pick a small number of portals from $H$ along $Q$ such that these distances are also approximated well with shortest paths in $H$ from $u$ to $Q$ through these portals.
\begin{lemma}\label{Lem:Portal}
Let $Q$ be a shortest path in an edge-weighted undirected graph $G = (V,E)$ and let $u\in V$ and $\epsilon',\epsilon'' > 0$ be given. Let $H$ be an undirected graph with $u\in V(H)\subseteq V$ such that for any $v_1,v_2\in V(H)$, $d_H(v_1,v_2)\geq d_G(v_1,v_2)$. Assume that for any $v\in V(Q)$, there is a $v'\in V(H)\cap V(Q)$ such that $d_H(u,v') + d_Q(v',v)\leq (1+\epsilon')d_G(u,v)$. Then there is a subset $P_H$ of $V(H)\cap V(Q)$ of size $O(1/\epsilon'')$ such that for any $v\in V(Q)$ there is a $p\in P_H$ such that $d_H(u,p) + d_Q(p,v)\leq (1+\epsilon')(1+\epsilon'')d_G(u,v)$.
\end{lemma}
\begin{proof}
The construction is similar to that of Thorup. The first portal $p_0$ added to $P_H$ is the vertex $v\in V(H)\cap V(Q)$ minimizing $d_H(u,v)$. Let $t$ be an endpoint of $Q$. We show how to construct $P_H\cap Q[p_0,t]$; the same construction is done for the other subpath of $Q$.

Let $p_j$ be the latest portal added to $P_H$ and traverse $Q[p_j,t]$ towards $t$ until encountering a vertex $p_{j+1}\in V(H)$ such that $d_H(u,p_j) + d_Q(p_j,p_{j+1}) > (1+\epsilon'')d_H(u,p_{j+1})$. Portal $p_{j+1}$ is then the next portal added to $P_H$. The process stops when reaching the vertex of $V(H)\cap V(Q)$ closest to $t$; this vertex is added as the final portal $p_k$ to $P_H$.

Let $v\in V(Q)$ be given. By assumption, there is a $v'\in V(H)\cap V(Q)$ such that $d_H(u,v') + d_Q(v',v)\leq (1+\epsilon')d_G(u,v)$. The above construction ensures that there is a $p\in P_H$ such that
\[
  d_H(u,p) + d_Q(p,v') + d_Q(v',v) \leq (1+\epsilon'')(d_H(u,v') + d_Q(v',v)) \leq (1+\epsilon')(1+\epsilon'')d_G(u,v).
\]

It remains to prove that $|P_H| = O(1/\epsilon'')$. By symmetry, it suffices to show $|P_H\cap Q[p_0,p_k]| = O(1/\epsilon'')$. For any $v\in V(H)\cap V(Q[p_0,p_k])$, define potential $\Phi(v) = d_H(u,v) + d_Q(v,p_k)$. For $j = 1,2,\ldots,k-1$,
\begin{align*}
\Phi(p_{j+1}) & = d_H(u,p_{j+1}) + d_Q(p_{j+1},p_k),\\
\Phi(p_j)    & = d_H(u,p_j) + d_Q(p_j,p_{j+1}) + d_Q(p_{j+1},p_k) > (1+\epsilon'')d_H(u,p_{j+1}) + d_Q(p_{j+1},p_k),
\end{align*}
so the potential is reduced by
\[
  \Phi(p_j) - \Phi(p_{j+1}) > \epsilon''d_H(u,p_{j+1})\geq \epsilon''d_H(u,p_0).
\]
Since $\Phi(p_0) = d_H(u,p_0) + d_Q(p_0,p_k)$ and $\Phi(p_k) = d_H(u,p_k)\geq d_Q(p_0,p_k) - d_H(u,p_0)$ (because $d_H(u,p_k) + d_H(u,p_0)\geq d_H(p_0,p_k)\geq d_G(p_0,p_k) = d_Q(p_0,p_k)$), we have $\Phi(p_0) - \Phi(p_k)\leq 2d_H(u,p_0)$, implying that $k = O(1/\epsilon'')$.
\end{proof}

\section{The First Phase}\label{sec:PhaseI}
Our data structure answers a query for vertex pair $(u,v)$ in two phases, \emph{Phase I} and \emph{Phase II}. In this section, we describe the preprocessing for Phase I and then the query part. The output and performance of Phase I applied to $u$ is stated in the following lemma ($v$ is symmetric). Phase I starts with $R_1(u)$ which is an arbitrary region $R$ (among at most two choices) such that $u\in\delta R$ and $u\notin\delta R'$ where $R'$ is the parent of $R$ in $\mathcal T$. Region $R_1(v)$ is defined similarly.
\begin{lemma}\label{Lem:PhaseI}
Phase I for vertex $u$ can be implemented to run in $O((\log\log n)^3/\epsilon^2)$ time using $O(n(\log\log n)^2/\epsilon)$ space. For the output $(P_u,\{\tilde d(u,p) | p\in P_u\})$, we have that for all $w\in V(\delta C_u)$, there is a $p\in P_u$ such that $d_{C_u}(u,w)\leq\tilde d(u,p) + d_{\delta C_u}(p,w)\leq(1+\epsilon/2)d_{C_u}(u,w)$ where $C_u$ is the child of $R_{uv} = \nca_{\mathcal T}(R_1(u),R_1(v))$ on the path in $\mathcal T$ from $R_{uv}$ to $R_u$. Furthermore, $|P_u| = O(1/\epsilon)$.
\end{lemma}
Note that any $uv$-path of $G$ must intersect $\delta C_u$. Phase I computes approximate distances to this separator but with the restriction that paths must be contained in $C_u$. Phase II is considered in Section~\ref{sec:PhaseII} and it extends the output of Phase I to approximate distances to $\delta_G C_u$ in the entire graph $G$.

\subsection{Preprocessing}
We start by constructing a recursive decomposition of $G$ and the associated decomposition tree $\mathcal T$. In order to traverse leaf-to-root paths of $\mathcal T$ efficiently, we set up a shortcutting system for $\mathcal T$. For any region $R\in\mathcal T$, let $i$ be the largest integer such that the depth of $R$ in $\mathcal T$ is divisible by $2^i$. For any integer $j$ between $0$ and $i$, we add a pointer from $R$ to the ancestor $R'$ $2^j$ levels above $R$. We refer to this pointer as a \emph{shortcut} and denote it by $R\rightarrow R'$. We can get from any region $R_1$ to any proper ancestor $R_2$ of $R_1$ by traversing only $O(\log\log n)$ shortcuts: first traverse the shortcut $R_1\rightarrow R'$ where $R'$ is the closest region to $R_2$ which is either $R_2$ or one of its descendants. Then recurse on pair $(R',R_2)$ until reaching $R_2$.


Before describing the preprocessing for Phase I, we need the following lemma. For any vertex $u\in V$, define $\mathcal R_u$ as the set of regions $R$ where $u\in\delta R\setminus C(R)$.
\begin{lemma}\label{Lem:RegionPath}
For all $u\in V$, regions of $\mathcal R_u$ form a subpath of a leaf-to-root path in $\mathcal T$.
\end{lemma}
\begin{proof}
Let $R$ be a region in $\mathcal R_u$. Non-corner vertices of shortest paths in $\mathcal P_R$ have degree three in $G$ (otherwise, their incident edges would have been merged into one in the construction of $R$) so $u$ must be incident to an interior edge $e$ of $R$. This edge cannot be a pseudo-edge since $u\notin C(R)$ but must be an edge of $G$. The same cannot be true for both child regions of $R$ in $\mathcal T$ since $G$ has degree three so at most one of these regions belongs to $\mathcal R_u$.
\end{proof}
We sometimes regard $\mathcal R_u$ as the subpath from the lemma.

For any shortcut $R_1\rightarrow R_2$, define $\delta(R_1,R_2)$ as the set of $O(1)$ vertices $u\in\delta R_1$ such that $u$ is the last vertex from $s$ on a path of $\mathcal P_{R_1}$ satisfying $u\in\delta R_2$; see Figure~\ref{fig:PhaseIData}.

For each shortcut $R_1\rightarrow R_2$ and each $w\in\delta(R_1,R_2)\cup(\delta R_1\setminus\delta R_2)$, we construct and store a size $O(1/\epsilon_1)$ $(w,R_2,1+\epsilon_1)$-portal set $P(w,R_1\rightarrow R_2)$ of $\delta R_2$ together with distances $d_{R_2}(w,p)$ for each $p\in P(w,R_1\rightarrow R_2)$ (Figure~\ref{fig:PhaseIData}); $\epsilon_1 = \Theta(\epsilon/\log\log n)$ will be specified precisely in Section~\ref{subsec:TimeStretch} below. This completes the description of the preprocessing.
\begin{figure}
\centerline{\scalebox{0.85}{\input{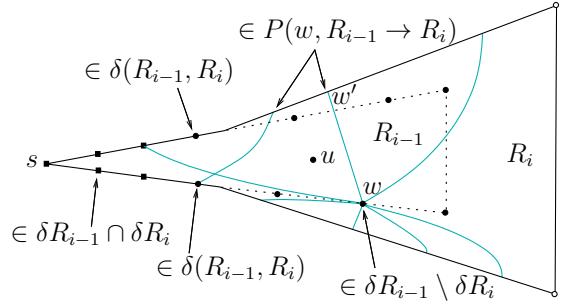}}}
\caption{Regions, $R_{i-1}$ and $R_i$ for shortcut $R_{i-1}\rightarrow R_i$. To simplify the figure, each region has only one hole and it is embedded on the infinite face. For every $w\in\delta(R_{i-1},R_i)\cup(\delta R_{i-1}\setminus\delta R_i)$ (circle vertices), a $(w,R_i,1+\epsilon_1)$-portal set $P(w,R_{i-1}\rightarrow R_i)$ is stored during preprocessing. In the $i$th iteration of the query algorithm, $H_i$ contains edges from $u$ to $P_{i-1}\subseteq V(\delta R_{i-1})$ and edges $(w,w')$ (one shown in figure) with $w'\in P(w,R_{i-1}\rightarrow R_i)\}$ for $w\in P_{i-1}$.}
\label{fig:PhaseIData}
\end{figure}
\begin{lemma}\label{Lem:TypeIPortalSet}
The total space required for Phase I is $O(n\log\log n/\epsilon_1)$.
\end{lemma}
\begin{proof}
It suffices to give an $O(n\log\log n/\epsilon_1)$ bound on the total size of portal sets defined above. Let $w\in V$ be given. By Lemma~\ref{Lem:RegionPath}, there can only be $O(\log\log n)$ shortcuts $R_1\rightarrow R_2$ where $R_1\in\mathcal R_w$ and $R_2\notin\mathcal R_w$. Hence there are only $O(\log\log n)$ shortcuts $R_1\rightarrow R_2$ where $w\in\delta R_1\setminus(\delta R_2\cup C(R_1))$. The total number of sets $\delta(R_1,R_2)$ and the total number of corners of $R_1$ over all shortcuts $R_1\rightarrow R_2$ is $O(n\log\log n)$ and $|\delta(R_1,R_2)| = O(1)$. As each set $P(w,R_1\rightarrow R_2)$ has size $O(1/\epsilon_1)$, total size of portal sets is $O(n\log\log n/\epsilon_1)$.
\end{proof}

\subsection{Query}\label{subsec:Query}
In this subsection, we present Phase I for query vertices $u$ and $v$. Figure~\ref{fig:PhaseIData} is useful to consult in the following. We assume that starting regions $R_1(u)$ and $R_1(v)$ are not on the same leaf-to-root path in $\mathcal T$. As we will see later, the other case is easily handled. 
Pseudocode for vertex $u$ can be found in Figure~\ref{fig:PhaseI}; the same call is made with $u$ replaced by $v$.
\begin{figure}[!ht]
\begin{tabbing}
\rule{\linewidth}{\arrayrulewidth}\\
d\=dd\=\quad\=\quad\=\quad\=\quad\=\quad\=\quad\=\quad\=\quad\=\quad\=\quad\=\quad\=\kill
\>\textbf{Phase I} for $u$:\\\\
\>1. \>\>let $R_1\rightarrow R_2\rightarrow\cdots\rightarrow R_k$ be the shortcuts from $R_1 = R_1(u)$ to $R_k = C_u$\\
\>2. \>\>let $P_2 = P(u,R_1\rightarrow R_2)$\\
\>3. \>\>for each $p\in P_2$, let $\tilde{d}_2(u,p) = d_{R_2}(u,p)$\\
\>4. \>\>for $i = 3$ to $k$\\
\>5. \>\>\>let $E_i = \{(u,w,\tilde{d}_{i-1}(u,w)) | w\in P_{i-1}\}$\\
\>6. \>\>\>let $E_i' = \{(w,w',d_{R_i}(w,w')) | w\in \delta(R_{i-1},R_i)\cup (P_{i-1}\setminus\delta R_i), w'\in P(w,R_{i-1}\rightarrow R_i)\}$\\
\>7. \>\>\>let $C_i$ be face $\delta R_i$ restricted to vertices that are either in $\delta(R_{i-1},R_i)\cup C(R_i)$ or\\
\>   \>\>\>are incident to edges in $E_i\cup E_i'$\\
\>8. \>\>\>construct the graph $H_i$ consisting of the edges $E_i\cup E_i'\cup E(C_i)$\\
\>9. \>\>\>for each $p\in V(H_i)$, let $\tilde{d}_i(u,p) = d_{H_i}(u,p)$\\
\>10.\>\>\>if $i < k$, construct $(u,H_i,1+\epsilon_1)$-portal set $P_i\subseteq V(H_i)$ of $\delta R_i$ of size $O(1/\epsilon_1)$\\
\>11.\>\>construct $(u,H_k,1+\epsilon_2)$-portal set $P_u\subseteq V(H_k)$ of $\delta R_k$ of size $O(1/\epsilon_2)$\\
\>12.\>\>output $(P_u,\{\tilde{d}(u,p)|p\in P_u\})$, where $\tilde{d}(u,p) = \tilde{d}_k(u,p)$\\
\rule{\linewidth}{\arrayrulewidth}
\end{tabbing}
\caption{Pseudocode for Phase I applied to $u$. Region $C_u$ is defined as in Lemma~\ref{Lem:PhaseI}.}\label{fig:PhaseI}
\end{figure}

Let $C_u$ be defined as in Lemma~\ref{Lem:PhaseI} ($C_v$ is defined similarly for $v$). Let $R_1\rightarrow R_2\rightarrow\cdots\rightarrow R_k$ denote the sequence of shortcuts from $R_1 = R_u$ to $R_k = C_u$. To simplify the code, we assume $k\geq 3$; the other case is straightforward. Note that $k = O(\log\log n)$.

In lines $2$ and $3$, we obtain the precomputed portal set $P_2 = P(u,R_1\rightarrow R_2)$ as well as distances $d_{R_2}(u,p)$ for each portal $p\in P_2$. Note that $P_2$ is well-defined by definition of $R_1(u)$.

In the $i$th iteration of the for-loop, we are given $P_{i-1}$ constituting portals for $\delta R_{i-1}$ and we form a graph $H_i$ containing $u$ and a subset of $V(\delta R_i)$ such that all distances from $u$ to $\delta R_i$ in $R_i$ can be approximated by going through $H_i$ and then along $\delta R_i$. An illustration of $H_i$ can be seen in Figure~\ref{fig:HiPhaseIData}.
\begin{figure}
\centerline{\scalebox{0.85}{\input{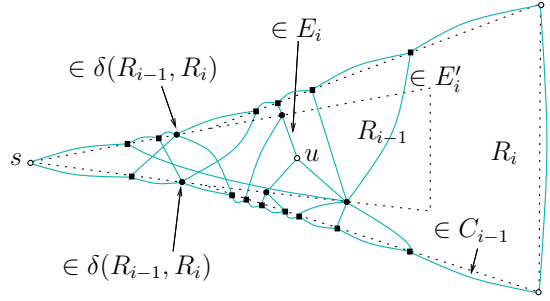}}}
\caption{Illustration of graph $H_i$ from Figure~\ref{fig:PhaseI}. Edges of $H_i$ are solid and boundaries $\delta R_{i-1}$ and $\delta R_i$ are dotted. Black circles are starting points of edges of $E_i'$, squares are their endpoints, and all other vertices of $H_i$ are white circles.}
\label{fig:HiPhaseIData}
\end{figure}

Edges of $E_i$ in line $5$ are added to $H_i$ and these represent the approximate distances to $P_{i-1}$ in $R_{i-1}$ found in the previous iteration. To have $H_i$ approximate distances from $u$ in $R_i$, we add another edge set $E_i'$, defined in line $6$; for every vertex in $w\in \delta(R_{i-1},R_i)\cup(P_{i-1}\setminus\delta R_i)$, we add to $H_i$ the edge $(w,w',d_{R_i}(w,w'))$ for every $w'\in P(w,R_{i-1}\rightarrow R_i)$, representing a shortest $ww'$-path in $R_i$. To allow $H_i$ to traverse $\delta R_i$, we add $C_i$ in line $7$. This is a compact representation of face $\delta R_i$ restricted to a subset of its vertices; each subpath of $\delta R_i$ between two consecutive vertices in this subset is a single edge of the same weight in $C_i$.

We show below that distances in $H_i$ from $u$ to $V(H_i)\cap\delta R_i$ approximate distances from $u$ to $\delta R_i$ in $R_i$. In order to avoid an explosion in the size of future portal sets, we form a $(u,H_i,1+\epsilon_1)$-portal set $P_i\subseteq V(H_i)$ of $\delta R_i$ of size only $O(1/\epsilon_1)$ in line $10$ which is then used in the next iteration.

Line $11$ is identical to line $10$ except that we use a value $\epsilon_2 > 0$ instead of $\epsilon_1$. We shall pick $\epsilon_2 \gg \epsilon_1$ which gives a much smaller portal set $P_u$ output in line $12$; this will help speed up Phase II. We do not use $\epsilon_2$ inside the for-loop in line $10$ since the approximation error builds up over each iteration so we need the smaller value $\epsilon_1$ there.

Lemma~\ref{Lem:PhaseI} follows from the following invariant for the for-loop in lines $4$--$10$:
\begin{description}
\item[Invariant:] At the start of the $i$th iteration of the for-loop in lines $4$--$10$ of Figure~\ref{fig:PhaseI}, for all $w\in V(\delta R_{i-1})$, there is a $p\in P_{i-1}$ such that $d_{R_{i-1}}(u,w)\leq \tilde{d}_{i-1}(u,p) + d_{\delta R_{i-1}}(p,w)\leq (1+\epsilon_1)^{2(i-1)}d_{R_{i-1}}(u,w)$, and $|P_{i-1}| = O(1/\epsilon_1)$.
\end{description}

Note that the invariant holds initially when $i = 3$ since $P_2$ is a $(u,R_2,1+\epsilon_1)$-portal set of $\delta R_2$ and its size is $O(1/\epsilon_1)$.

\paragraph{Maintenance of invariant:} Let $3\leq i < k$ be given and assume the invariant holds at the beginning of the $i$th iteration of the for-loop. We show that it also holds at the beginning of the $(i+1)$th iteration. Let $Q$ be a shortest path in $R_i$ from $u$ to a vertex $w_i\in\delta R_i$. We show that there is a $p_i\in V(H_i)\cap\delta R_i$ such that $d_{H_i}(u,p_i) + d_{\delta R_i}(p_i,w_i)$ approximates the weight of $Q$ up to a factor of $(1+\epsilon_1)^{2i-1}$. The second inequality of the invariant will then follow from Lemma~\ref{Lem:Portal}. The first inequality follows since $\tilde d_i(u,p)$ is the weight of an actual path in $R_i$ from $u$ to $p$ for each $p\in P_i$.

Let $w_{i-1}$ be the last vertex on $Q$ such that the subpath of $Q$ from $u$ to $w_{i-1}$ is contained in $R_{i-1}$. In particular, $w_{i-1}\in\delta R_{i-1}$. By the invariant at the beginning of the $i$th iteration, there is a portal $p_{i-1}\in P_{i-1}$ such that $\tilde{d}_{i-1}(u,p_{i-1}) + d_{\delta R_{i-1}}(p_{i-1},w_{i-1}) \leq (1+\epsilon_1)^{2(i-1)}d_{R_{i-1}}(u,w_{i-1})$.

Assume first that $w_{i-1}\in (\delta R_{i-1}\cap\delta R_i)\setminus\delta(R_{i-1},R_i)$. Then $Q$ is contained in $R_{i-1}$ and $H_i$ approximates the distance from $u$ to $w_i$ up to a factor of $(1+\epsilon_1)^{2(i-1)}$ with the path starting with $(u,p_{i-1})\in E_i$ and followed by edges of $C_i$.

Now assume that $w_{i-1}\in (\delta R_{i-1}\setminus\delta R_i)\cup\delta(R_{i-1},R_i)$. Consider first the case where
$p_{i-1}\in\delta R_{i-1}\setminus\delta R_i$ (Figure~\ref{fig:PhaseIData} with $w$ playing the role of $p_{i-1}$). We have the precomputed portal set $P(p_{i-1},R_{i-1}\rightarrow R_i)$ and it contains a portal $p_i$ such that $d_{R_i}(p_{i-1},p_i) + d_{\delta R_i}(p_i,w_i) \leq (1+\epsilon_1)d_{R_i}(p_{i-1},w_i)$. Hence,
\begin{align*}
  d_{H_i}(u,p_i) + d_{\delta R_i}(p_i,w_i) & \leq \tilde{d}_{i-1}(u,p_{i-1}) + d_{R_i}(p_{i-1},p_i) + d_{\delta R_i}(p_i,w_i)\\
  & \leq (1+\epsilon_1)(\tilde{d}_{i-1}(u,p_{i-1}) + d_{R_i}(p_{i-1},w_i))\\
  & \leq (1+\epsilon_1)(\tilde{d}_{i-1}(u,p_{i-1}) + d_{\delta R_{i-1}}(p_{i-1},w_{i-1}) + d_{R_i}(w_{i-1},w_i))\\
  & \leq (1+\epsilon_1)((1+\epsilon_1)^{2(i-1)}d_{R_{i-1}}(u,w_{i-1}) + d_{R_i}(w_{i-1},w_i))\\
  & \leq (1+\epsilon_1)^{2i-1}d_{R_i}(u,w_i).
\end{align*}

Now consider the other case where $w_{i-1}\in (\delta R_{i-1}\setminus\delta R_i)\cup\delta(R_{i-1},R_i)$ and $p_{i-1}\in\delta R_{i-1}\cap\delta R_i$. Then a shortest path from $p_{i-1}$ to $w_{i-1}$ in $\delta R_{i-1}$ contains a vertex $p_{i-1}'\in\delta(R_{i-1},R_i)$ and the subpath from $p_{i-1}$ to $p_{i-1}'$ is contained in $C_i\subseteq H_i$. We have a precomputed portal set $P(p_{i-1}',R_{i-1}\rightarrow R_i)$ containing a portal $p_i$ such that $d_{R_i}(p_{i-1}',p_i) + d_{\delta R_i}(p_i,w_i) \leq (1+\epsilon_1)d_{R_i}(p_{i-1}',w_i)$. This gives
\begin{align*}
  d_{H_i}(u,p_i) + d_{\delta R_i}(p_i,w_i) & \leq d_{R_{i-1}}(u,p_{i-1}) + d_{\delta R_{i-1}}(p_{i-1},p_{i-1}') + d_{R_i}(p_{i-1}',p_i) + d_{\delta R_i}(p_i,w_i)\\
  & \leq (1+\epsilon_1)(d_{R_{i-1}}(u,p_{i-1}) + d_{\delta R_{i-1}}(p_{i-1},p_{i-1}') + d_{R_i}(p_{i-1}',w_i))\\
  & \leq (1+\epsilon_1)(d_{R_{i-1}}(u,p_{i-1}) + d_{\delta R_{i-1}}(p_{i-1},p_{i-1}') +{}\\
  & \phantom{{}\leq  (1+\epsilon_1)(}d_{\delta R_{i-1}}(p_{i-1}',w_{i-1}) + d_{R_i}(w_{i-1},w_i))\\
  & =    (1+\epsilon_1)(d_{R_{i-1}}(u,p_{i-1}) + d_{\delta R_{i-1}}(p_{i-1},w_{i-1}) + d_{R_i}(w_{i-1},w_i))\\
  & \leq (1+\epsilon_1)((1+\epsilon_1)^{2(i-1)}d_{R_{i-1}}(u,w_{i-1}) + d_{R_i}(w_{i-1},w_i))\\
  & \leq (1+\epsilon_1)^{2i-1}d_{R_i}(u,w_i),
\end{align*}
as desired.

\paragraph{Termination:} The invariant shows that at the beginning of the $k$th iteration, for all $w\in V(\delta R_{k-1})$, there is a $p\in P_{k-1}$ such that $d_{R_{k-1}}(u,w)\leq \tilde{d}_{k-1}(u,p) + d_{\delta R_{k-1}}(p,w)\leq (1+\epsilon_1)^{2(k-1)}d_{R_{k-1}}(u,w)$. Line $11$ is identical to line $10$ for $i = k$ except that $\epsilon_1$ is replaced by $\epsilon_2$ so line $11$ increases the approximation by a factor of $(1+\epsilon_1)(1+\epsilon_2)$. Below we choose $\epsilon_1$ and $\epsilon_2$ such that $(1+\epsilon_1)^{2k-1}(1+\epsilon_2)\leq 1 + \epsilon/2$. This will imply Lemma~\ref{Lem:PhaseI}.


\subsection{Bounding query time and stretch}\label{subsec:TimeStretch}
Obtaining shortcuts in line $1$ can be done in $O(k) = O(\log\log n)$ time. Lines $2$ and $3$ take $O(|P_2|) = O(1/\epsilon_1)$ time as $P(u,R_1\rightarrow R_2)$ and distances $d_{R_2}(u,p)$ for $p\in P_2$ are precomputed.

We first show how a single iteration $i$ of the for-loop in lines $4$--$10$ can be implemented to run in $O((\log(1/\epsilon_1))/\epsilon_1^2)$ time. Then we show how to improve it to $O(1/\epsilon_1^2)$. Finding $E_i$ takes $O(|P_{i-1}|) = O(1/\epsilon_1)$ time. As $|\delta(R_{i-1},R_i)| = O(1)$, $E_i'$ can be found in time $O((1 + |P_{i-1}|)/\epsilon_1) = O(1/\epsilon_1^2)$. Face $C_i$ is obtained in $O((|\delta(R_{i-1},R_i)\cup C(R_i)| + |E_i\cup E_i'|)\log(|\delta(R_{i-1},R_i)\cup C(R_i)| + |E_i\cup E_i'|)) = O(\log(1/\epsilon_1)/\epsilon_1^2)$ time by sorting the vertices according to their cyclic ordering in an Euler tour of face $\delta R_i$. Graph $H_i$ and distances $\tilde{d}_i(u,p)$ for $p\in V(H_i)$ can then be obtained in $O(|E_i| + |E_i'| + |C_{i-1}|) = O(1/\epsilon_1^2)$ time. To obtain $P_i$, apply the portal construction algorithm in the proof of Lemma~\ref{Lem:Portal} to each shortest path in $\mathcal P_i$ restricted to $H_i$. This takes $O(|\mathcal P_i||V(H_i)|) = O(|V(H_i)|) = O(1/\epsilon_1^2)$ time. Line $11$ takes $O(1/\epsilon_2)$ time.

We improve the time bound to $O(1/\epsilon_1^2)$ by avoiding the sorting step above. Instead, we omit adding edges of $C_i$ to $H_i$ and apply a variant of the portal construction algorithm in the proof of Lemma~\ref{Lem:Portal} to each shortest path $Q_i\in\mathcal P_i$ restricted to $H_i$. We describe this variant in the following.

First observe that $V(H_i)\cap Q_i$ can be partitioned into $O(1/\epsilon_1)$ subsets $A_i(1),\ldots,A_i(\ell)$ where each subset $A_i(j)$ is either a singleton set consisting of a vertex in $\delta(R_{i-1},R_i)\cup C(R_i)$, a singleton set consisting of an endpoint of an edge in $E_i$, or a set $P(w,R_{i-1}\rightarrow R_i)\cap Q_i$ for each $w\in\delta(R_{i-1},R_i)\cup (P_{i-1}\setminus\delta R_i)$. Assume that each set $A_i(j)$ is sorted along $Q_i$ in non-decreasing distance from the source $s$ in shortest path tree $T$; this sorting can be done in the preprocessing step.

The first portal $p_0$ to be added to $P_i$ is the vertex $p\in V(H_i)\cap Q_i$ minimizing $d_{H_i}(u,p)$. We can identify $p_0$ in $O(|V(H_i)\cap Q_i|) = O(1/\epsilon_1^2)$ time. Let $t$ be the vertex of $Q_i$ farthest from $s$. As in the proof of Lemma~\ref{Lem:Portal}, we only describe the algorithm for adding portals to $Q_i[p_0,t]$; adding portals along $Q_i[s,p_0]$ is symmetric.

For each set $A_i(j)$, we keep a pointer to the first vertex in its sorted order. We then make a single pass over the sets $A_i(j)$ and for each of them move its pointer forward to the first vertex in $Q_i[p_0,t]$ (if any) for which the distance to it cannot be approximated by going through the previously added portal $p_0$, i.e., the first vertex $w\in Q_i[p_0,t]$ such that $d_{H_i}(u,p_0) + d_{Q_i}(p_0,w) > (1+\epsilon_1)d_{H_i}(u,w)$. Among the vertices with pointers to them, the one closest to $p_0$ in $Q_i$ is then added to $P_i$ as the next portal, and $p_0$ is updated to this vertex. Additional passes are made until the pointers have moved past all vertices of their respective $A_i(j)$ sets.

Correctness follows since the set of portals formed is the same as that obtained by the portal-construction algorithm of Thorup. Running time is $O(1/\epsilon_1^2)$. To see this, note that each pass (except possibly the last) adds at least one portal to $P_i$ so the number of passes is $O(1/\epsilon_1)$. Furthermore, each pass takes $O(\ell + x) = O(1/\epsilon_1 + x)$ time where $x$ is the total number of vertices visited in that pass over all sets $A_i(j)$. Since the total number of vertices visited over all passes is $O(|\cup_{j = 1}^{\ell}A_i(j)|) = O(|V(H_i)\cap Q_i|) = O(1/\epsilon_1^2)$, the time bound follows.

We can obtain a stretch of $1 + \epsilon/2$ for the approximate distances obtained in the final iteration of Phase I as follows. Since $(1+\epsilon_1)^{2k}\leq e^{2k\epsilon_1} < 1/(1 - 2k\epsilon_1)$ when $2k\epsilon_1 < 1$, we pick $\epsilon_1 = \epsilon/(8k)$ to obtain $(1+\epsilon_1)^{2k} < 1/(1 - \epsilon/4)$ which is at most $1 + \epsilon/3$ when $\epsilon\leq 1$. Picking $\epsilon_2 = \epsilon/8$ gives
$(1+\epsilon_1)^{2k-1}(1+\epsilon_2) < 1+\epsilon/2$ for $\epsilon$ smaller than some positive constant. This shows Lemma~\ref{Lem:PhaseI}.

\section{The Second Phase}\label{sec:PhaseII}
Phase II takes as input the sets $P_u$ and $P_v$ with associated approximate distances $\tilde d(u,p)$ and $\tilde d(v,p)$ that were output by Phase I. The output of Phase II has the properties stated in the following lemma. Denote by $S_{uv}$ the shortest path separator in $G$ that separates $R_{uv}$ into $C_u$ and $C_v$. In Section~\ref{sec:ApproxDist}, we efficiently obtain from this output an approximate $uv$-distance.
\begin{lemma}\label{Lem:PhaseII}
Phase II for $u$ can be implemented to run in $O((\log\log n)^2/\epsilon + (\log\log n)/\epsilon^2)$ time using $O(n\log\log n/\epsilon^2)$ space, given the output from Phase I. For the output $(V(H),\{d_H(u,p) | p\in V(H)\}$ from Phase II, we have $|V(H)| = O(\log\log n/\epsilon^2)$ and for any $w\in V(S_{uv})$, there is a vertex $p\in V(H)$ such that $d_G(u,w)\leq d_H(u,p) + d_T(p,w)\leq (1+\epsilon)d_G(u,w)$.
\end{lemma}

\subsection{Preprocessing}
The preprocessing for Phase II consists of the following four steps:

\paragraph{Step $1$:} For each $w\in V$ and each of the at most two regions $R_w$ with $w\in\delta R_w$ and $w\notin\delta R_w'$ where $R_w'$ is the parent of $R_w$ in $\mathcal T$, we form a $(w,G,1+\epsilon_2)$-portal set $P(w)'$ of $\delta_G R_w$. From this we form and store a subset $P(w)$ of $V(\delta R_w)$. This subset contains $P(w)'\cap V(\delta R_w)$. In addition, for every $p\in P(w)'$ and every left or right side $Q\in\mathcal P(R_w)$ (see Section~\ref{subsec:RecDecomp}) such that $Q$ and $p$ are contained in the same path of $T$, $P(w)$ contains the successor and predecessor (if any) of $p$ on $Q$; Figure~\ref{fig:Pw} gives an illustration.
\begin{figure}
\centerline{\scalebox{0.85}{\input{Pw.pstex_t}}}
\caption{Step $1$ of the preprocessing for Phase II. White vertices belong to $\delta R_w$ and black and white vertices belong to $\delta_G R_w$. Edges of $G$ incident to $\delta_G R_w$ and edges of path $Q$ are solid. The predecessor and successor on $Q$ of a vertex of $P(w)'$ belong to $P(w)$.}
\label{fig:Pw}
\end{figure}
Note that $V(Q)\subseteq V(\delta R_w)$ so $P(w)\subseteq V(\delta R_w)$. For any $v\in V(\delta R_w)\subseteq V(\delta_G R_w)$ there is a $p'\in P(w)'$ with $d_G(w,p') + d_{\delta_G R_w}(p',v)\leq (1+\epsilon_2)d_G(w,v)$. Since $v\in V(\delta R_w)$ the shortest path in $\delta_G R_w$ from $p'$ to $v$ intersects $V(\delta R_w)$. Hence there is a $p\in P(w)\cup C(R_w)$ which is either $p'$, the successor or predecessor of $p'$ on $Q$, or a corner in $C(R_w)$ such that $d_G(w,p) + d_{\delta R_w}(p,v)\leq (1+\epsilon_2)d_G(w,v)$. We also have $|P(w)| = O(1/\epsilon_2)$.


\paragraph{Step $2$:} For each shortcut $R_1\rightarrow R_2$ and each $u\in\delta R_1\setminus\delta R_2$, store a $(u,G,1+\epsilon_2)$-portal set $P_1(u,R_1\rightarrow R_2)$ of $\delta_G R_1$ of size $O(1/\epsilon_2)$; Call this a \emph{type 1} portal set.

\paragraph{Step $3$:} For any shortcut $R_1\rightarrow R_2$, \emph{dual portal set} $P^*(R_1,R_2)$ is the set of vertices $p^*\in(\delta R_1\cap\delta R_2)\setminus C(R_1)$ for which a vertex $w$ exists with $R_w\in R_1\leadsto R_2$, $R_w\neq R_2$, such that $p^*\in P(w)$; see Figure~\ref{fig:PhaseIIData}.
Define $\overline P^*(R_1,R_2) = P^*(R_1,R_2)\cup C(R_1)\cup\delta(R_1,R_2)$. Note that $\overline P^*(R_1,R_2)\subseteq V(\delta R_1)$. For each $p^*\in \overline P^*(R_1,R_2)$, store a $(p^*,G,1+\epsilon_2)$-portal set $P_2(p^*,R_1\rightarrow R_2)$ of $\delta_G R_1$ of size $O(1/\epsilon_2)$ together with distances $d_G(p^*,p')$ for each $p'\in V(\delta_G(R_1))$. Refer to it as a \emph{type 2} portal set.

The definition of type 2 portal sets is rather technical so let us give the high-level idea for introducing them; see Figure~\ref{fig:PhaseIIData} in the following. As in Phase I, we jump along shortcuts $R_{i-1}\rightarrow R_i$ in Phase II; a detailed description is given in the next subsection. In Phase II, we need approximate distances in $G$ from certain portals $p$ in $\delta R_{i-1}$ to $\delta_G R_{i-1}$. However, $p$ might also be present in $\delta R_i$. In this case, we cannot afford to associate portal sets with $p$ and shortcut $R_{i-1}\rightarrow R_i$ since $p$ may occur in several regions of $\mathcal T$ (see Lemma~\ref{Lem:RegionPath}). However, vertices $w$ for which $R_w$ is sandwiched in between $R_{i-1}$ and $R_i$ can pay for dual portal set $P^*(R_{i-1},R_i)$ and the associated type 2 portal sets. As we show below, we can obtain an approximate distance from $p$ to any $w'\in\delta_G R_{i-1}$ by first going along $T$ from $p$ to a nearby $p^*\in P^*(R_{i-1},R_i)$, then along a shortest path in $G$ from $p^*$ to a portal $p'\in\delta_G R_{i-1}$ in the type 2 portal set of $p^*$, and finally from $p'$ to $w'$ along $T$.

Using hashing, we can access each type 1 and type 2 portal set in $O(1)$ time from the vertex and the shortcut defining it.

\paragraph{Step $4$:} For each shortcut $R_1\rightarrow R_2$ and for any shortest path $Q\in\mathcal P_{R_1}$, we keep a vEB-tree, allowing us to find the successor/predecessor of any vertex of $Q$ in the subset $V(Q)\cap\overline P^*(R_1,R_2)$ in $O(\log\log n)$ time. With hashing, space required for the vEB-tree is $O(|\overline P^*(R_1,R_2)|)$~\cite{Predecessor,RangeQueries}. As mentioned in~\cite{Predecessor}, both space and query bounds can be made deterministic.
\begin{figure}
\centerline{\scalebox{0.85}{\input{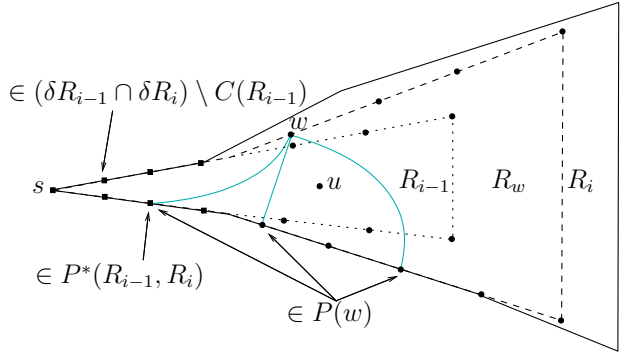}}}
\caption{Region $R_w$ sandwiched in between $R_{i-1}$ and $R_i$ for shortcut $R_{i-1}\rightarrow R_i$. Curves from $w\in\delta R_w$ end in $P(w)$. If a vertex is in $(\delta R_{i-1}\cap\delta R_i)\setminus C(R_{i-1})$ (square vertices) and in $P(w)$ then it is in $P^*(R_{i-1},R_i)$.}
\label{fig:PhaseIIData}
\end{figure}
\begin{lemma}\label{Lem:TypeIIAndIIIPortalSet}
The total space required for Phase II is $O(n\log\log n/\epsilon_2^2)$.
\end{lemma}
\begin{proof}
The total size of portal sets $P(u)$ over all $u\in V$ is $O(n/\epsilon_2)$. A proof similar to that of Lemma~\ref{Lem:TypeIPortalSet} shows that the total size of all type $1$ portal sets is $O(n\log\log n/\epsilon_2)$.

To bound the size of type 2 portal sets, consider two shortcuts $R_1\rightarrow R_2$ and $R_1'\rightarrow R_2'$. By Lemma~\ref{Lem:RegionPath}, $P^*(R_1,R_2)$ and $P^*(R_1',R_2')$ are disjoint if $R_1\leadsto R_2$ and $R_1'\leadsto R_2'$ are not contained in the same leaf-to-root path $P$ in $\mathcal T$. For any $w\in V$, there are only $O(\log\log n)$ shortcuts $R_1''\rightarrow R_2''$ where $R_w\in R_1''\leadsto R_2''\subseteq P$. Hence $|P(w)\cap P^*(R_1,R_2)|$ summed over all shortcuts $R_1\rightarrow R_2$ is $O(|P(w)|\log\log n) = O(\log\log n/\epsilon_2)$. Hence, the total size of all dual portal sets $P^*(R_1,R_2)$, and hence also the total size of all vEB trees, is $O(n\log\log n/\epsilon_2)$. Recall that the number of regions is $O(n)$. Each region has $O(1)$ corners and each set $\delta(R_1,R_2)$ has constant size so the total size of all sets $\overline P^*(R_1,R_2)$ is $O(n\log\log n/\epsilon_2)$. Each element of these sets has a type 2 portal set of size $O(1/\epsilon_2)$.
\end{proof}

\subsection{Query}
Let $\mathcal P_u = C_u\leadsto G$ be the path from $C_u$ to the root $G$ of $\mathcal T$. In the following, for any vertex $w$, denote by $R_w$ the region (if any) such that $R_w\in\mathcal P_u$, $w\in\delta R_w$, and $w\notin\delta R_w'$ where $R_w'$ is the parent of $R_w$ in $\mathcal T$. Observe that $V(S_{uv})\subseteq\cup_{R\in\mathcal P_u}V(\delta R)$. Phase II for $u$ takes the output from Phase I and produces output satisfying Lemma~\ref{Lem:PhaseII}.
\begin{figure}[!ht]
\begin{tabbing}
\rule{\linewidth}{\arrayrulewidth}\\
d\=dd\=\quad\=\quad\=\quad\=\quad\=\quad\=\quad\=\quad\=\quad\=\quad\=\quad\=\quad\=\kill
\>\textbf{Phase II} for $u$:\\\\
\>1. \>\>let $R_1\rightarrow R_2\rightarrow\cdots\rightarrow R_k$ be the shortcuts from $R_1 = C_u$ to $R_k = G$\\
\>2. \>\>let $P_1 = P_u$ (portal set output from Phase I)\\
\>3. \>\>let $H$ be the graph initially consisting of edges $(u,p,\tilde{d}(u,p))$ for all $p\in P_1$\\
\>4. \>\>for $i = 2$ to $k$\\
\>5. \>\>\>add to $H$ edges $(p,q,d_G(p,q))$ for all $p\in P_{i-1}\setminus V(\delta R_i)$ and $q\in P_1(p,R_{i-1}\rightarrow R_i)$\\
\>6. \>\>\>for each $p\in P_{i-1}$ and each $p^*\in\{\succ(p,\overline P^*(R_{i-1},R_i)),\pred(p,\overline P^*(R_{i-1},R_i))\}$\\
\>7. \>\>\>\>add to $H$ edge $(p,q,d_T(p,p^*)+d_G(p^*,q))$ for all $q\in P_2(p^*,R_{i-1}\rightarrow R_i)$\\
\>8. \>\>\>let $P_i = P_{i-1}\cap V(\delta R_i)$\\
\>9. \>\>output $(V(H),\{d_H(u,p) | p\in V(H)\})$\\
\rule{\linewidth}{\arrayrulewidth}
\end{tabbing}
\caption{Pseudocode for Phase II applied to $u$. Here, $C_u$ resp.~$P_u$ denotes the final region resp.~portal set reached in Phase I and $S_{uv}$ is the shortest path separator in $G$ that separates $R_{uv}$ into $C_u$ and $C_v$. In line $6$, $\succ(p,\overline P^*(R_{i-1},R_i))$ resp.~$\pred(p,\overline P^*(R_{i-1},R_i))$ refers to the successor resp.~predecessor of $u$ in $\overline P^*(R_{i-1},R_i)$.}\label{fig:PhaseII}
\end{figure}
We give a high-level description of Phase II before going into details.  Pseudocode can be seen in Figure~\ref{fig:PhaseII}. The algorithm traverses shortcuts $R_1\rightarrow R_2\rightarrow\cdots\rightarrow R_k$ from $R_1 = C_u$ to the root $R_k = G$ of $\mathcal T$ and incrementally constructs a graph $H$ which at termination will satisfy Lemma~\ref{Lem:PhaseII}. In line $3$, edges of $H$ represent approximate paths found in Phase I. These paths correspond to subpaths of the final full paths in $G$ (corresponding to the final $H$) and the subpaths are prefixes of these full paths that are contained in $C_u$. Consider the $i$th iteration of the for-loop. In line $5$, we check if any subpath endpoint $p\in P_{i-1}$ disappears as a boundary vertex when jumping from $R_{i-1}$ to $R_i$. If so, we can extend the subpath to full paths $u\leadsto p\leadsto q$ for each $q\in P_1(p,R_{i-1}\rightarrow R_i)$. The other interesting case is when $p\in\delta R_{i-1}\cap\delta R_i$. Then we do not have a type $1$ portal set associated with $p$ and $R_{i-1}\rightarrow R_i$ but it might be that some separator vertices $w$ of $S_{uv}$ that are present in $\delta_G R_{i-1}$ are no longer present in $\delta R_i$ and we need to ensure that there is a good path in $H\cup S_{uv}$ from $p$ to $w$. This case is handled in lines $6$ and $7$ where we ensure such a good path from $p$ to $w$ by using the type $2$ portal sets associated with vertices of $\overline P^*(R_{i-1},R_i))$ that are close to $p$.

To show correctness, i.e., that the set output in line $9$ satisfies Lemma~\ref{Lem:PhaseII}, let $w$ be any vertex on $S_{uv}$ and let $P$ be a shortest path in $G$ from $u$ to $w$. Let $w'$ be the last vertex on $P$ such that $P[u,w']$ is contained in $C_u$. Note that $w'\in V(\delta C_u)$. By Lemma~\ref{Lem:PhaseI}, there is a $p_1\in P_1$ such that $d_{C_u}(u,w')\leq\tilde d(u,p_1) + d_{\delta C_u}(p_1,w')\leq(1+\epsilon/2)d_{C_u}(u,w')$. Since $V(S_{uv})\subseteq\cup_{R\in\mathcal P_u}V(\delta R)$, we have $R_w\in\mathcal P_u$. Since $p_1\in P_1\subseteq V(\delta C_u)$, we also have $R_{p_1}\in\mathcal P_u$. Let $R_{i_w-1}\rightarrow R_{i_w}$ and $R_{i_{p_1}-1}\rightarrow R_{i_{p_1}}$ be the shortcuts such that $R_w\in R_{i_w-1}\leadsto R_{i_w}$, $R_w\neq R_{i_w}$, and $R_{p_1}\in R_{i_{p_1}-1}\leadsto R_{i_{p_1}}$, $R_{p_1}\neq R_{i_{p_1}}$. We consider two cases in the following: $i_{p_1}\leq i_w$ and $i_{p_1} > i_w$.

\paragraph{Case $1$, $i_{p_1}\leq i_w$ (Figure~\ref{fig:Case1}):}
\begin{figure}
\centerline{\scalebox{0.85}{\input{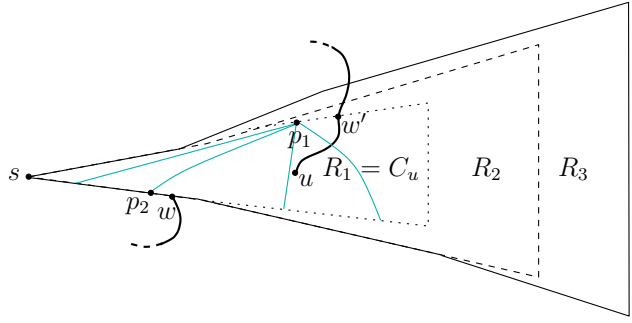}}}
\caption{Case $1$ in the correctness proof for Phase II; here $i_{p_1} = 1$ and $i_w = 3$. The first and last part of a shortest $uv$-path in $G$ through $w'$ is shown. Curves from $p_1$ end in vertices of $P_1(p_1,R_{i-1}\rightarrow R_i)$.}
\label{fig:Case1}
\end{figure}
Consider iteration $i = i_{p_1}$ of the for-loop. Since $p_1\in V(\delta C_u)\cap V(\delta R_{i-1})$, we have $p_1\in P_{i-1}$. Since $R_{p_1}\neq R_{i_{p_1}}$, it follows that $p_1\notin V(\delta R_i)$. Hence, the final $H$ contains an edge $(p_1,p,d_G(p_1,p))$ for each portal $p\in P_1(p_1,R_{i-1}\rightarrow R_i)$ (line $5$). Since $w\in V(S_{uv})\subseteq V(\delta_G C_u)$ and since $i\leq i_w$, we must have $w\in\delta_G R_{i-1}$ so there is a portal $p_2\in P_1(p_1,R_{i-1}\rightarrow R_i)$ satisfying $d_G(p_1,p_2) + d_{\delta_G R_{i-1}}(p_2,w)\leq (1+\epsilon_2)d_G(p_1,w)$. Note that $d_{\delta_G R_{i-1}}(p_2,w) = d_T(p_2,w)$. The path in the final graph $H$ consisting of edges $(u,p_1)$ and $(p_1,p_2)$ followed by the path in $T$ from $p_2$ to $w$ has weight at most
\begin{align*}
\tilde{d}(u,p_1) + d_G(p_1,p_2) + d_T(p_2,w) & \leq (1+\epsilon_2)(\tilde{d}(u,p_1) + d_G(p_1,w))\\
 & \leq (1+\epsilon_2)(\tilde{d}(u,p_1) + d_{\delta C_u}(p_1,w') + d_G(w',w))\\
 & \leq (1+\epsilon_2)((1+\epsilon/2)d_{C_u}(u,w') + d_G(w',w))\\
 & \leq (1+\epsilon_2)(1+\epsilon/2)d_G(u,w).
\end{align*}
\paragraph{Case $2$, $i_{p_1} > i_w$ (Figure~\ref{fig:Case2}):}
\begin{figure}
\centerline{\scalebox{0.85}{\input{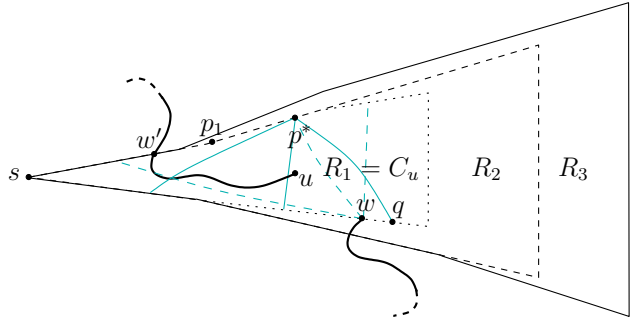}}}
\caption{Case $2$ in the correctness proof for Phase II; here $i_{p_1} = 2$ and $i_w = 1$. The first and last part of a shortest $uv$-path in $G$ through $w'$ is shown. Dashed curves from $w$ end in vertices of $P(w)$ one of which is $q^*\in\overline{P}^*(R_{i-1},R_i)$. Solid curves from $p^*$ end in $P_2(p^*,R_{i-1}\rightarrow R_i)$.}
\label{fig:Case2}
\end{figure}
Consider iteration $i = i_w$ of the for-loop. Since $p_1\in V(\delta R_w)$, there is a portal $p_w^*\in P(w)\cup C(R_w)$ such that $d_G(w,p_w^*) + d_T(p_w^*,p_1)\leq (1+\epsilon_2)d_G(w,p_1)$ (see step $1$ of the preprocessing). If $p_w^*\in\delta R_{i-1}\cap\delta R_i$ then $p_w^*\in P^*(R_{i-1},R_i)\cup C(R_{i-1})$. Otherwise, there is a vertex $p'\in\delta(R_{i-1},R_i)$ such that the path in $T$ from $p_1$ to $p_w^*$ contains $p'$. Hence,
\begin{align*}
d_G(w,p') + d_T(p',p_1) & \leq d_G(w,p_w^*) + d_T(p_w^*,p') + d_T(p',p_1) = d_G(w,p_w^*) + d_T(p_w^*,p_1)\\
 & \leq (1+\epsilon_2)d_G(w,p_1).
\end{align*}

It follows from the above that there is a vertex $q^*\in\overline{P}^*(R_{i-1},R_i)$ such that $d_G(w,q^*) + d_T(q^*,p_1)\leq (1+\epsilon_2)d_G(w,p_1)$. Thus, for one of the two choices of $p^*$ in line $6$, we have
\[
  d_G(w,p^*) + d_T(p^*,p_1) \leq d_G(w,q^*) + d_T(q^*,p_1) \leq (1+\epsilon_2)d_G(w,p_1).
\]
For that choice of $p^*$, let $q$ be the portal in $P_2(p^*,R_{i-1}\rightarrow R_i)$ such that $d_G(p^*,q) + d_T(q,w)\leq (1+\epsilon_2)d_G(p^*,w)$. The path in $H$ consisting of edges $(u,p_1)$ (added in line $3$) and $(p_1,q)$ (added in line $7$) followed by the path in $T$ from $q$ to $w$ has weight at most
\begin{align*}
\tilde{d}(u,p_1) + d_T(p_1,p^*) + d_G(p^*,q) + d_T(q,w) & \leq (1+\epsilon_2)(\tilde{d}(u,p_1) + d_T(p_1,p^*) + d_G(p^*,w))\\
 & \leq (1+\epsilon_2)(\tilde{d}(u,p_1) + (1+\epsilon_2)d_G(p_1,w))\\
 & \leq (1+\epsilon_2)^2(\tilde{d}(u,p_1) + d_{\delta C_u}(p_1,w') + d_G(w',w))\\
 & \leq (1+\epsilon_2)^2((1+\epsilon/2)d_{C_u}(u,w') + d_G(w',w))\\
 & \leq (1+\epsilon_2)^2(1+\epsilon/2)d_G(u,w).
\end{align*}

To show Lemma~\ref{Lem:PhaseII}, note that $(1+\epsilon_2)^2(1+\epsilon/2) = (1+\epsilon/8)^2(1+\epsilon/2)$ is at most $1+\epsilon$ for any $\epsilon$ less than some positive constant. As for query time, lines $1$ and $2$ can be executed in $O(k) = O(\log\log n)$ time. By Lemma~\ref{Lem:PhaseI}, adding edges to $H$ in line $3$ can be done in $O(|P_1|) = O(1/\epsilon)$ time. The total time to find successors and predecessors over all iterations of the for-loop is $O(|P_1|k\log\log n) = O((\log\log n)^2/\epsilon)$. The additional time spent in the for-loop is bounded by the number of edges added to $H$. The total number of edges added in line $7$ is $O(|P_1|k/\epsilon_2) = O(\log\log n/\epsilon^2)$. Note that every $p$ considered in line $5$ is not included in $P_i$ in line $8$. Hence, we add a total of $O(|P_1|/\epsilon_2) = O(1/\epsilon^2)$ edges in line $5$.

\section{Obtaining the approximate distance}\label{sec:ApproxDist}
In this section, we show how to obtain an approximate $uv$-distance within the time and space stated in Theorem~\ref{Thm:Main}, given the output of Phase II.

We execute Phase II for both $u$ and $v$, getting outputs $(V(H_u),\{d_{H_u}(u,p) | p\in V(H_u)\}$ and $(V(H_v),\{d_{H_v}(u,p) | p\in V(H_v)\}$, respectively. To find an approximate $uv$-distance, assume first that $V(H_u)\cup V(H_v)\subseteq V(S_{uv})$. Let $Q$ be one of the two shortest paths from $s$ in $T$ bounding $S_{uv}$ and sort the vertices of $V(H_u)\cap Q$ along $Q$. To do the sorting efficiently, we make a DFS traversal of $T$ during preprocessing and label each vertex with an integer time stamp denoting when it was first visited in the traversal. Now, sorting the vertices of $V(H_u)\cap Q$ along $Q$ corresponds to integer sorting their precomputed labels. With the algorithm of Han and Thorup~\cite{IntegerSorting}, this takes $O(|V(H_u)|\sqrt{\log\log(|V(H_u)|)})$ time. We then remove all $w\in V(H_u)\cap Q$ for which there is another $w'\in V(H_u)\cap Q$ such that $d_{H_u}(u,w') + d_Q(w',w) \leq d_{H_u}(u,w)$. This can be done in $O(|V(H_u)|)$ time with two linear scans over $V(H_u)\cap Q$, one in sorted order and the other in reverse sorted order; the first resp.~second scan removes $w$ if there is a $w'$ before resp.~after $w$ in the order considered such that the inequality holds. Let $V_{u,Q}$ be the resulting subset and form a similar subset $V_{v,Q}$ of $V(H_v)$.

Let $\tilde{d}_Q(u,v)$ be the minimum of $d_{H_u}(u,p_u) + d_Q(p_u,p_v) + d_{H_v}(p_v,v)$ over all pairs $(p_u,p_v)\in V_{u,Q}\times V_{v,Q}$ where $Q[p_u,p_v]$ has no interior vertices belonging to $V_{u,Q}\cup V_{v,Q}$. This takes $O(|V(H_u)|)$ time. Compute a similar value $\tilde{d}_{Q'}(u,v)$ for the other shortest path $Q'$ of $T$ bounding $S_{uv}$. The approximate distance output is $\tilde{d}(u,v) = \min\{\tilde{d}_Q(u,v),\tilde{d}_{Q'}(u,v)\}$.

We need to show $\tilde{d}(u,v)\leq (1+\epsilon)d_G(u,v)$. Pick $w\in V(S_{uv})$ such that $d_G(u,w) + d_G(w,v) = d_G(u,v)$ and assume that $w\in Q$; the case where $w\in Q'$ is symmetric. By Lemma~\ref{Lem:PhaseII}, there are vertices $p_u\in V_{u,Q}$ and $p_v\in V_{v,Q}$ such that $d_{H_u}(u,p_u) + d_Q(p_u,w) \leq (1+\epsilon)d_G(u,w)$ and $d_{H_v}(v,p_v) + d_Q(p_v,w)\leq (1+\epsilon)d_G(v,w)$. Pick a pair $(p_u',p_v')\in V_{u,Q}\times V_{v,Q}$ belonging to $Q[p_u,p_v]$ such that they occur on this path in the order $p_u\leadsto p_u'\leadsto p_v'\leadsto p_v$ and such that $Q[p_u',p_v']$ has no interior vertices belonging to $V_{u,Q}\cup V_{v,Q}$. Then
\begin{align*}
\tilde{d}_Q(u,v) & \leq d_{H_u}(u,p_u') + d_Q(p_u',p_v') + d_{H_v}(p_v',v)\\
                 & \leq d_{H_u}(u,p_u) + d_Q(p_u,p_u') + d_Q(p_u',p_v') + d_Q(p_v',p_v) + d_{H_v}(p_v,v)\\
                 & =    d_{H_u}(u,p_u) + d_Q(p_u,p_v) + d_{H_v}(p_v,v)\\
                 & \leq d_{H_u}(u,p_u) + d_Q(p_u,w) + d_Q(w,p_v) + d_{H_v}(p_v,v)\\
                 & \leq (1+\epsilon)(d_G(u,w) + d_G(w,v))\\
                 & =    (1+\epsilon)d_G(u,v),
\end{align*}
showing the desired. Above, we assumed that $V(H_u)\cup V(V_v)\subseteq V(S_{uv})$. If this is not the case, we modify $V(H_u)$ as follows. Partition $V(H_u)\setminus V(S_{uv})$ into maximal-size groups where in each group $M$, all vertices $w$ have the same nearest neighbor $w'$ in $T$ belonging to $S_{uv}$. Replace vertices of $M$ by $w'$ in $V(H_u)$ and instead of approximate distances $d_{H_u}(u,w)$ for $w\in M$, use instead $d_{H_u}(u,w') := \min_{w\in M\cup(\{w'\}\cap V(H_u))} d_{H_u}(u,w) + d_T(w,w')$ for the approximate distance for $w'$. A similar update is done to $V(H_v)$, ensuring that $V(H_u)\cup V(V_v)\subseteq V(S_{uv})$. It is easy to see that the above analysis still carries through.

We have shown Theorem~\ref{Thm:Main} in the case where $R_u$ and $R_v$ are not on the same leaf-to-root path in $\mathcal T$. If instead, say, $R_v = \nca_{\mathcal T}(R_u,R_v)$ then $v\in S_{uv}$ and Phase I and II for $u$ gives a portal set of $S_{uv}$. Our algorithm above is modified to find the portal $p$ nearest to $v$ on $S_{uv}$ and outputs $d_{H_u}(u,p) + d_T(p,v)$, giving the desired stretch. This shows Theorem~\ref{Thm:Main} in the remaining case where $R_u$ and $R_v$ are on the same leaf-to-root path.

\section{Concluding Remarks}\label{sec:ConclRem}
We gave a $(1+\epsilon)$-approximate distance oracle for undirected $n$-vertex planar graphs and fixed $\epsilon > 0$ with $O(n(\log\log n)^2)$ space and $O((\log\log n)^3)$ query time which improves the previous best query time-space product from $O(n\log n)$ to $O(n(\log\log n)^5)$.

We have not focused on preprocessing time. With a simple implementation, we should get near-quadratic preprocessing time and it is possible that the exact space-efficient oracle in~\cite{ExactOraclePlanarMozesSommer} can speed this up further to $\tilde O(n^{3/2})$ as the number of precomputed distances required by our oracle is only $\tilde O(n)$. With techniques from, e.g.,~\cite{OraclePlanarKlein, OraclePlanarThorup}, we can likely get down to $\tilde O(n)$.

The dependency on $\epsilon$ in the query time-space product is slightly worse; it is roughly $1/\epsilon^4$ ($1/\epsilon^3$ when $1/\epsilon = O(\log\log n)$) compared to $1/\epsilon^2$ in~\cite{OraclePlanarKlein, OraclePlanarThorup} and roughly $1/\epsilon$ in~\cite{CompactOraclesPlanar} (where the latter has a slightly worse dependency on $n$ than~\cite{OraclePlanarKlein, OraclePlanarThorup}). Using mainly Monge properties, we believe it should be possible to replace at least one $1/\epsilon$ factor by $\log(1/\epsilon)$. Getting $o(\log\log n)$ query time and $O(n(\log\log n)^c)$ space for some constant $c$ seems problematic with our techniques due to the $\Theta(\log\log n)$ bottleneck from the use of vEB trees when answering queries.

Extension to planar digraphs seems promising due to similarities between our structure and that for digraphs in~\cite{OraclePlanarThorup}. Extension to minor-free graphs would also be interesting.

\end{document}